\documentclass[12pt]{article}
\usepackage{graphicx,amsmath,amssymb,url,enumerate,mathrsfs,epsfig,color}
\usepackage{tikz}
\usepackage{float}

\usepackage{amsmath}
\usepackage{amssymb}
\usepackage{yfonts}
\usepackage{caption}

\usepackage{tikz,pgfplots}

\usetikzlibrary{calc, patterns,arrows.meta, arrows, shapes.geometric}

\usepackage{graphicx,amsmath,amssymb,url,enumerate,mathrsfs,epsfig,color}

\usetikzlibrary{decorations.text}
\usetikzlibrary{decorations.markings}

\pgfplotsset{compat=1.8}

\usepackage{enumerate}
\usepackage{boxedminipage}
\usepackage{makeidx}
\usepackage{multicol}
\usepackage{xcolor}
\usepackage{mathtools}
\usepackage{comment}

\usetikzlibrary{calc, patterns,arrows, shapes.geometric}

\newenvironment{proof}{\noindent{\bf Proof~:}}{\QED\medskip}
\def\QED{\hskip0.1em\hfill\null\ \null\nobreak\hfill
\kern3pt\lower1.8pt\vbox{\hrule\hbox
{\vrule\kern1pt\vbox{\kern1.7pt \hbox{$\scriptstyle
QED$}\kern0.2pt}\kern1pt\vrule}\hrule}}

\newtheorem{prop}{Proposition}[section]

\newtheorem{rem}{Remark}[section]
\newtheorem{definition}{Definition}[section]

\newtheorem{example}{Example}[section]
\title{Uniformity theory of weighted composites}
\author{Marcelo Epstein\footnote{University of Calgary, Calgary, Canada (Corresponding Author) \newline email: mepstein@ucalgary.ca} ~ and Mohammadjavad Javadi\footnote{University of Alberta, Edmonton, Canada \newline email: javadisi@ualberta.ca}}   
\date{}
\begin{document}
\maketitle

\begin{abstract} By including spatially varying volume fractions of the constituents, composites are obtained that are akin to functionally graded media. An extension of the mathematical apparatus of double groupoids is proposed to incorporate these materials with the aim of eventually classifying defects of misalignment and their time evolution.

\end{abstract}
\bigskip


\section{Introduction}

The theory of material composites can be regarded, in principle, as a subdiscipline of the general theory of mixtures \cite{bowen1, bowen2}. In this much more circumscribed work, on the other hand, we intend to emphasize the fact that each constituent of a composite must, by definition, satisfy a condition of {\it material uniformity} in the sense of Noll \cite{noll}. Since a composite made of several uniform constituents is not necessarily uniform, a continuous distribution of {\it defects of misalignment} \cite{jmps} arises in general, which may initiate a process of {\it material evolution} of the composite through a possible remodelling of the individual constituents, or through the onset of chemical reactions, or other such time-dependent events. Examples of the first kind \cite{eande, book} implicitly assume that the volumetric fractions of each of the constituents remain constant in space and time. In removing these limitations, we introduce the concept of {\it weighted uniformity} which allows us to expand the mathematical apparatus introduced in \cite{jmps} for the analysis of defects of misalignment.

The representation of binary composites in terms of double groupoids was first proposed in \cite{symmetry} and \cite{jmps}. Further developments were presented in \cite{victor}, and a generalization to $n$-ary composites in terms of $n$-fold groupoids was suggested in \cite{book}. The extension presented herein for the analysis of binary composites of varying volume fractions entails a slight generalization of the notion of material double groupoid. To make the presentation as self-contained as possible, Section \ref{sec:materialgroupoids} revisits in some detail the notion of material groupoid. After motivating an application to composite materials, Section \ref{sec:materialdouble} offers a review of the mathematical notion of double groupoid and its application to binary composites. In Section \ref{sec:volumefractions}, the bundle of volume fractions of a binary composite is introduced, and the notions of weighted groupoids and weighted double groupoids are defined by a systematic extension of their base manifold to a fibre bundle. Within this framework, an informal geometric description of material evolution in composites is provided, both in terms of material remodelling and temporal changes in volume fractions. Section \ref{sec:nonslim} considers the possibility of making use of non-slim double groupoids \cite{natale} via a possible application to uniform composites.

\section{Material groupoids}

\label{sec:materialgroupoids}

\subsection{Brief review of the notion of material groupoid}

Let the mechanical response of an elastic solid body $\mathcal B$ be described in some reference configuration by a constitutive equation of the form
\begin{equation}  \label{eq1}
{\bf T}= {\bf f}({\bf F}, X),
\end{equation}
where $\bf f$ is a symmetric-tensor-valued function of the deformation gradient tensor $\bf F$ and of the material point $X \in {\mathcal B}$. The value of $\bf f$ is interpreted as the Cauchy stress tensor $\bf T$. Two points $X_1, X_2 \in {\mathcal B}$ are made of the same material if they are {\it materially isomorphic} \cite{noll}, namely, if there is a non-singular linear map ${\bf P}_{12}$, called a {\it material isomorphism},\footnote{Technically speaking, this is an isomorphism between the tangent spaces to the manifold $\mathcal B$ at $X_1$ and $X_2$. Each such isomorphism is a member of the general linear group $GL(3,{\mathbb R})$.} such that
\begin{equation}
{\bf f}({\bf F},X_2)={\bf f}({\bf F\,P}_{12},X_1),\;\;\;\;\;\;\;\;\;\;\forall\,{\bf F} \in GL(3,{\mathbb R}).
\end{equation} 
For the purpose of conceptual visualization we will interpret a material isomorphism as an arrow with source $X_1$ and target $X_2$, as shown in Figure \ref{fig1}. A material automorphism (with $X_2=X_1$) is clearly a material symmetry ${\bf G}_1$ at $X_1$. In the suggested visual representation, a material symmetry is a loop-shaped arrow.
\begin{figure}[H]
\begin{center}
\begin{tikzpicture}[scale = 1]
\draw[thick] (0,0) ellipse (3cm and 2cm);
\draw[o-stealth'] (-1,0) to [bend left] (1.5,1);
\node[above] at  (0.5,1) {${\bf P}_{12}$};
\draw[-stealth'] (-1,0) arc(90:-260:0.4);
\node[below] at (-1.,-0.75) {${\bf G}_1$};
\node[above left] at (-1,0) {$X_1$};
\node [right] at (1.5,1) {$X_2$};
\node at (2.5,-1.5) {$\mathcal B$};

\end{tikzpicture}
\caption{Material isomorphisms as arrows}
\label{fig1}
\end{center}
\end{figure}
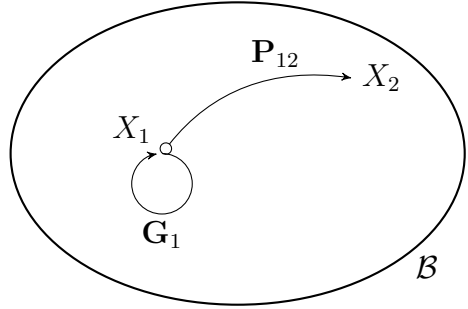

Consider the collection $\mathcal Z$ of all material isomorphisms induced by a given constitutive equation on $\mathcal B$. This collection is never empty, since the loop-shaped arrow corresponding to the trivial (unit) material symmetry at each point $X \in  \mathcal B$ is always available. Moreover, if ${\bf P}_{12}$ is a material isomorphism from $X_1$ to $X_2$, it follows that the inverse map ${\bf P}^{-1}_{12}$ is a material isomorphism from $X_2$ to $X_1$. Finally, if ${\bf P}_{12}$ is a material isomorphism from $X_1$ to $X_2$ and ${\bf P}_{23}$ is a material isomorphism from $X_2$ to $X_3$, the product ${\bf P}_{23}\,{\bf P}_{12}$ is a material isomorphism between $X_1$ and $X_3$. Indeed, the property of being made of the same material is an equivalence relation. In the graphical representation we may say that arrows can be composed if they are related in a tip-to-tail fashion, as shown in Figure \ref{fig2}.
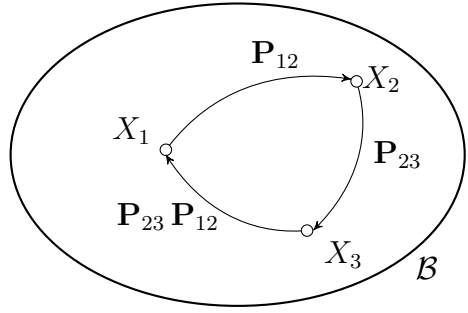
\begin{figure}[H]
\begin{center}
\begin{tikzpicture}[scale = 1]
\draw[thick] (0,0) ellipse (3cm and 2cm);
\draw[o-stealth'] (-1,0) to [bend left] (1.5,1);
\node[above] at  (0.5,1) {${\bf P}_{12}$};
\node[above left] at (-1,0) {$X_1$};
\node [right] at (1.5,1) {$X_2$};
\draw[o-stealth'] (1.55,1.05) to [bend left] (1,-1);
\node[right] at (1.65,0) {${\bf P}_{23}$};
\node[below right] at (1,-1) {$X_3$};
\draw[o-stealth'] (1,-1) to [bend left] (-0.95,0);
\node[below left] at (-0.1,-0.5) {${\bf P}_{23}\,{\bf P}_{12}$};
\node at (2.5,-1.5) {$\mathcal B$};
\end{tikzpicture}
\caption{Composition of tip-to-tail as arrows}
\label{fig2}
\end{center}
\end{figure}
Informally speaking, the set of arrows $\mathcal Z$ defined over the set of objects (material points) $\mathcal B$, with the composition operation just described, constitutes the {\it material groupoid} associated with the body $\mathcal B$ and the given constitutive equation.

More generally, the concept of {\it groupoid}, as introduced a hundred years ago by Brandt \cite{brandt}, is a generalization of the notion of group. It encompasses, as demonstrated in the case of the material groupoid, the local symmetries at each element of a system with the distant symmetries between different elements \cite{weinstein}.

\begin{definition} {\rm  A {\it groupoid} ${\mathcal Z} \rightrightarrows {\mathcal B}$ consists of two sets, $\mathcal Z$ and $\mathcal B$, called, respectively, the set of {\it arrows} and the set of {\it objects}, related by two surjective {\it projection maps}, $\alpha, \beta: {\mathcal Z} \to {\mathcal B}$, called the {\it source} map and the {\it target} map, respectively. Just as in the case of a group, an associative operation (the {\it product} or {\it composition} of arrows) is defined in $\mathcal Z$. This operation, however, can only be carried out between two arrows $u, v \in {\mathcal Z}$ if the ``tip-to-tail'' condition
\begin{equation}
\beta(u)=\alpha(v)
\end{equation}
is satisfied. The result, denoted by $vu$, satisfies the conditions
\begin{equation}
\alpha(vu)=\alpha(u),\;\;\;\;\;\;\;\;\;\; \beta(vu)=\beta(v).
\end{equation}
For each object $b \in {\mathcal B}$ a unique {\it identity arrow} $\epsilon(b)$ exists such that $\epsilon(u)\,u=u$, for every $u \in {\mathcal Z}$ such that $\beta(u)=b$, and $ u\,\epsilon(u)=u$, for every  $u \in {\mathcal Z}$ such that $\alpha(u)=u$. For each arrow $u \in {\mathcal Z}$ there is a (unique) {\it inverse arrow} $u^{-1} \in {\mathcal Z}$ such that $u^{-1}\,u=\epsilon(\alpha(u))$, and $u\,u^{-1}=\epsilon(\beta(u))$.
}
\end{definition}

\begin{rem} {\rm In spite of its length, the definition of a groupoid is a rather intuitive rendering of the main features of examples such as the set of material isomorphisms of a body, or the set of operations that preserve the appearance of a tiled floor (including the mutual exchange of congruent tiles). A groupoid can also be regarded as a generalization of the concept of equivalence relation, a generalization that takes into account that an object may be equivalent to another (including itself) in a multiplicity of ways.
}
\end{rem}

\begin{definition} {\rm A groupoid is {\it transitive} if there is at least one arrow between every pair of objects}
\end{definition}

\begin{example} {\rm A body all of whose points are made of the same material is said to be {\it uniform}. Consequently, a body is uniform if, and only if, its material groupoid is transitive.}
\end{example}

At each point $b \in {\mathcal B}$, the collection of all the arrows $u$ of a groupoid ${\mathcal Z} \rightrightarrows {\mathcal B}$ such that $\alpha(u)=\beta(u)=b$ forms a group ${\mathcal G}_b$ called the {\it vertex group} at $b$. It is not difficult to show that if $u$ is an arrow from $b_1=\alpha(u)$ to $b_2=\beta(u)$ then the corresponding vertex groups ${\mathcal G}_1$ and ${\mathcal G}_2$ are conjugate, that is,
\begin{equation}
{\mathcal G}_2=u\,{\mathcal G}_1\, u^{-1}.
\end{equation}
In the case of the material groupoid this equation translates into the evident statement that if two points of a body are made of the same material their material symmetry groups are conjugate (the conjugation arising from the fact that in a given reference configuration the points may not be in the same reference state).

\section{Material double groupoids}
\label{sec:materialdouble}

\subsection{Composites and weighted uniformity}
\label{sec:weighted}

\begin{definition}{\rm An elastic body $\mathcal B$ is a {\it weighted n-ary elastic composite} (or simply a {\it composite}) if the following conditions are satisfied:
\begin{enumerate}
\item The constitutive equation can be expressed as a sum
\begin{equation}
{\bf T}=\sum_{\alpha=1}^n  k_\alpha(X)\;{\bf f}_\alpha ({\bf F}, X),
\end{equation}
where each of the coefficient scalar functions $k_\alpha$ is positive definite and
\begin{equation}
\sum_{\alpha=1}^n k_\alpha = 1.
\end{equation}
\item Each of the functions ${\bf f}_\alpha \;(\alpha=1,...,n)$ satisfies the {\it uniformity} property
\begin{equation}
 {\bf f}_\alpha({\bf F}, X_2) = {\bf f}_\alpha({\bf F\, P}^\alpha_{12}, X_1),\;\;\;\;\;\;\;\;\forall \; {\bf F} \in GL(3,{\mathbb R}),
\end{equation}
for each pair points $X_1, X_2 \in {\mathcal B}$, where ${\bf P}^\alpha_{12}$ is a material isomorphism between $X_1$ and $X_2$ associated with the function ${\bf f}_\alpha$.
\end{enumerate}
}
\end{definition}

A composite is thus made up from $n$ uniform {\it constituents}, each of which contributes additively to the total constitutive response in proportion to its {\it volume fraction} $k_\alpha$.

\begin{definition}\label{def:weightedunif} {\rm A composite $\mathcal B$ is said to enjoy the property of {\it weighted uniformity} if for every pair of points $X_1, X_2 \in {\mathcal B}$, the intersection of the sets of material isomorphisms between these points of all the constituents is non-empty, namely
\begin{equation}
\bigcap\limits_{\alpha=1}^n {\mathcal P}^\alpha_{12}  \ne \emptyset,
\end{equation}
where ${\mathcal P}^\alpha_{12}$ denotes the collection of all material isomorphisms ${\bf P}^\alpha_{12}$ between $X_1$ and $X_2$ in the constituent $\alpha$.
}
\end{definition}

\begin{rem} {\rm If each of the volume fractions $k_\alpha$ is spatially constant, weighted uniformity implies uniformity, as defined in \cite{jmps}).
}
\end{rem}

Weighted uniformity does not imply ordinary uniformity, since by having spatially variable volume fractions the resulting body can be regarded as a {\it (functionally) graded material}. What weighted uniformity does imply is that there is no misalignment between the various constituents. Moreover, the symmetry groups of the composite constitutive function turn out to be mutually conjugate. Such materials, regardless of whether they are composites or not, are called {\it unisymmetric} and have been studied in \cite{eandmdel}. Weighted uniformity on the other hand, confers to the total composite material additional properties not envisioned in that study of the more general case of functionally graded materials.

\subsection{Double groupoids}
\label{sec:double}

A faithful mathematical representation of weighted composites can be achieved by means of $n$-fold groupoids \cite{book}. The case of binary weighted composites (with $n=2$) corresponds to {\it double groupoids}, which have been studied in greater detail \cite{brownmac, natale}. The emergence of the concept of double groupoid in the context of solid mechanics is to to be expected upon realizing that in the case of binary composites there are two underlying material groupoids, one for each constituent.\footnote{For historical reasons, these groupoids are called, respectively, horizontal and vertical, a designation that is not amenable to be extended to more than two constituents. Nevertheless, we will retain it for the binary case.} These two groupoids share the same set of objects $\mathcal B$, but each of them has a different set of arrows, $\mathcal H$ and $\mathcal V$, and their corresponding projection maps $\alpha_H, \beta_H: {\mathcal H} \to {\mathcal B}$ and $\alpha_V, \beta_V: {\mathcal V} \to {\mathcal B}$. A double groupoid can be thought of as a superstructure $\mathcal Q$ that encompasses the two constituent groupoids. 
\begin{definition} {\rm A {\it double groupoid} is a set ${\mathcal Q}$ equipped with 2 groupoid structures
whose bases are themselves ordinary groupoids ${\mathcal H}$ and ${\mathcal V}$, with a common base $\mathcal B$. Denoting the compositions of the first and second groupoid structures of $\mathcal Q$ by $\odot_V$ and $\odot_H$, respectively,
the following {\it compatibility condition} must be satisfied
\begin{equation} \label{dg5again}
({\mathbb A}\odot_H{\mathbb B})\odot_V({\mathbb C}\odot_H{\mathbb D})=({\mathbb A}\odot_V{\mathbb C})\odot_H({\mathbb B}\odot_V{\mathbb D})\;\;\;\;\;\;\;\;\;\;\;\;\forall\,{\mathbb A},{\mathbb B},{\mathbb C},{\mathbb D} \in {\mathcal Q}
\end{equation}
whenever the operations are possible. Other conditions of compatibility between the two structures can be summarized by stating that the structure maps of each groupoid structure of $\mathcal Q$ are groupoid morphisms with respect to the other structure.\footnote{These properties will be described below more explicitly. A groupoid morphism is a map (or {\it functor}) that assigns to each arrow of one groupoid an arrow of the other groupoid, and that preserves identities and compositions.}

}
\end{definition}

A double groupoid $\mathcal Q,$ as just defined, is endowed with two groupoids structures, called the horizontal structure and the vertical structure, as shown in the diagram of Equation (\ref{eq:doublegroupoid}). The horizontal structure has as its set of objects the vertical groupoid $\mathcal V$, while the vertical structure has as its objects the horizontal groupoid $\mathcal H$. Consider first the horizontal structure of $\mathcal Q$, and let ${\tilde \alpha}_H, {\tilde \beta}_H: {\mathcal Q}\to {\mathcal V}$ denote its source and target maps. The source ${\tilde \alpha}_H({\mathbb A})$ of an element $\mathbb A$ of $\mathcal Q$ is thus an arrow in $\mathcal V$, and so is its target ${\tilde \beta}_H({\mathbb A})$. In the vertical structure, on the other hand, with projections ${\tilde \alpha}_V, {\tilde \beta}_V: {\mathcal Q}\to {\mathcal H}$, the same element ${\mathbb A} \in {\mathcal Q}$ will have as its source and target two arrows of $\mathcal H$, ${\tilde \alpha}_H({\mathbb A})$ and ${\tilde \beta}_B({\mathbb A})$. For this reason, the elements of $\mathcal Q$ are called {\it squares}, as they project on two pairs of arrows, one pair from each of the constituent groupoids. Moreover, these four arrows are consistent in the sense that further projection of these arrows on the common set of objects $\mathcal B$ results in four corner points. For example, the left lower corner corresponding to the square $\mathbb A$ is obtained as $\alpha_V({\tilde \alpha}_H({\mathbb A}))=\alpha_H({\tilde \alpha}_V({\mathbb A}))$.

\begin{equation} \label{eq:doublegroupoid}
\begin{tikzpicture}[ scale=1,baseline=(current  bounding  box.east)]]
\tikzset{->/.style={decoration={
  markings,
  mark=at position .98 with {\arrow{stealth'}}},postaction={decorate}}}
\node at (0,0) {$\mathcal H$};
\node at (0,3) {$\mathcal Q$};
\node at (3,0) {$\mathcal B$};
\node at (3,3) {$\mathcal V$};
\draw[-stealth'] (0.5,0.1)--(2.5,0.1);
\node[above] at (1.5,0.1) {$\alpha_H$};
\draw[-stealth'] (0.5,-0.1)--(2.5,-0.1);
\node[below] at (1.5,-0.1) {$\beta_H$};
\draw[-stealth'] (0.5,3.1)--(2.5,3.1);
\node[above] at (1.5,3.1) {$\tilde{\alpha}_H$};
\draw[-stealth'] (0.5,2.9)--(2.5,2.9);
\node[below] at (1.5,2.9) {$\tilde{\beta}_H$};
\draw[-stealth'] (-0.1,2.5)--(-0.1,0.5);
\node[left] at (-0.1,1.5) {$\tilde{\alpha}_V$};
\draw[-stealth'] (0.1,2.5)--(0.1,0.5);
\node[right] at (0.1,1.5) {$\tilde{\beta}_V$};
\draw[-stealth'] (2.9,2.5)--(2.9,0.5);
\node[left] at (2.9,1.5) {${\alpha}_V$};
\draw[-stealth'] (3.1,2.5)--(3.1,0.5);
\node[right] at (3.1,1.5) {${\beta}_V$};
\end{tikzpicture}
\end{equation}

To complete the description of the horizontal and vertical groupoid structures, it is necessary to define the respective composition laws, units, and inverses. Let $\mathbb A$ be a square in $\mathcal Q$, as represented in (\ref{dg2}). Its two horizontal sides $s$ and $t$ are obtained as source and target of $\mathbb A$ using the projections of the vertical structure, that is, with the maps $\tilde{\alpha}_V$ and $\tilde{\beta}_V$, respectively. Similarly, the vertical sides $\hat s$ and $\hat t$ are the source and target projections of $\mathbb A$ according to the horizontal-structure projections $\tilde{\alpha}_H$ and $\tilde{\beta}_H$.

\begin{equation} \label{dg2}
\begin{tikzpicture}[scale=0.75, baseline=(current  bounding  box.center)]
  \draw[thick,-stealth'] (0,0) -- (3,0);
  \draw[thick,-stealth'] (0,3) -- (3,3);
  \draw[thick,-stealth'] (3,0) -- (3,3);
  \draw[thick,-stealth'] (0,0) -- (0,3);
\node at (1.5,1.5) {$\mathbb A$};
\node at (4.5,1.5) {${\hat t}={\tilde \beta}_H({\mathbb A})$};
\node at (-1.5,1.5) {${\hat s}={\tilde \alpha}_H({\mathbb A})$};
\node[above] at (1.5,3) {$t={\tilde \beta}_V({\mathbb A})$};
\node[below] at (1.5,0) {$s={\tilde \alpha}_V({\mathbb A})$};
\end{tikzpicture}
\end{equation}
Since $\mathcal Q$ is endowed with two different groupoid structures, each structure must have its own well-defined product, which will be denoted by $\odot_H$ and $\odot_V$, respectively, for the horizontal and vertical structures. In the horizontal structure, the product of two squares, ${\mathbb A} \odot_H {\mathbb A}'$, can only be carried out if ${\hat s}={\hat t}'$, in agreement with the usual concatenation criterion in a groupoid.

The composition law\index{double groupoid!compositions of squares} of the horizontal structure of $\mathcal Q$ is obtained by applying the composition law of the horizontal side groupoid $\mathcal H$ to each pair of adjacent horizontal arrows of the two squares. This operation is conveniently represented graphically as
\begin{equation} \label{dg3}
\begin{tikzpicture}[baseline=(current  bounding  box.center)]
\foreach \x in {0,1.5,4.5}
{  \draw[thick,stealth'-] (1.5+\x,0) -- (0+\x,0);
  \draw[thick,stealth'-] (1.5+\x,1.5) -- (0+\x,1.5);
  \draw[thick,-stealth'] (1.5+\x,0) -- (1.5+\x,1.5);
  \draw[thick,-stealth'] (0+\x,0) -- (0+\x,1.5);}
\node[left] at(-0.,0.75) {${\mathbb A}\odot_H {\mathbb A}'\;=\;{\hat s}'$};
\node[above] at(0.75,1.5) {$  t'$};
\node[above] at(2.25,1.5) {$ t$};
\node[below] at(0.75,-0.1) {$  s'$};
\node[below] at(2.25,0) {$ s$};
\node[right] at (3,0.75) {${\hat t}$};
\node at (3.75,0.75) {$=$};
\node[left] at (4.5,0.75) {${\hat s'}$};
\node[above] at (5.25,1.5) {$tt'$};
\node[below] at (5.25,0) {$ss'$};
\node[right] at (6,0.75) {${\hat t}$};
\end{tikzpicture}
\end{equation}
The unit square ${\tilde\epsilon}_H({\hat u})$ of the horizontal structure at the vertical arrow $\hat u$ is the square
\begin{equation} \label{dg3a}
\begin{tikzpicture}[baseline=(current  bounding  box.center)]
\foreach \x in {0}
{  \draw[thick,stealth'-] (1.5+\x,0) -- (0+\x,0);
  \draw[thick,stealth'-] (1.5+\x,1.5) -- (0+\x,1.5);
  \draw[thick,-stealth'] (1.5+\x,0) -- (1.5+\x,1.5);
  \draw[thick,-stealth'] (0+\x,0) -- (0+\x,1.5);}
\node[left] at(-0.,0.75) {${\hat u}$};
\node[above] at(0.75,1.5) {$ \epsilon_H({\beta}_V({\hat u}))$};
\node[below] at(0.75,-0.1) {$ \epsilon_H({\alpha}_V({\hat u}))$};
\node[right] at (1.5,0.75) {${\hat u}$};
\end{tikzpicture}
\end{equation}

Similarly, in the vertical structure the product ${\mathbb A} \odot_V{\mathbb A}'$ can only be carried out if $s=t'$, or graphically
\begin{equation} \label{dg4}
\begin{tikzpicture}[baseline=(current  bounding  box.center)]
\foreach \x in {0,-1.5}
{  \draw[thick,stealth'-] (1.5,0+\x) -- (0,0+\x);
  \draw[thick,stealth'-] (1.5,1.5+\x) -- (0,1.5+\x);
  \draw[thick,-stealth'] (1.5,0+\x) -- (1.5,1.5+\x);
  \draw[thick,-stealth'] (0,0+\x) -- (0,1.5+\x);}
\node[left] at(-0.,0) {${\mathbb A}\odot_V {\mathbb A}'\;=\;\;\;\;\;$};
\node[left] at (0,0.75) {${\hat s}$};
\node[above] at(0.75,1.5) {$  t$};
\node[right] at (1.5,0.75) {${\hat t}$};
\node[right] at (1.5,-0.75) {${\hat t}'$};
\node[below] at(0.75,-1.5) {$s'$};
\node[left] at (0,-0.75) {$\hat s'$};
\node[right] at (2.,0.1) {$=$};
\node[left] at (3.5,0.1) {$ {\hat s} {\hat s}'$};
\foreach \x in {3.5}
{\draw[thick,stealth'-] (1.5+\x,-0.75) -- (0+\x,0-0.75);
  \draw[thick,stealth'-] (1.5+\x,1.5-0.75) -- (0+\x,1.5-0.75);
  \draw[thick,-stealth'] (1.5+\x,0-0.75) -- (1.5+\x,1.5-0.75);
  \draw[thick,-stealth'] (0+\x,0-0.75) -- (0+\x,1.5-0.75);}
\node[right] at (5,0.1) {${\hat t}{\hat t}'$};
\node[above] at (4.25,0.75) {$t$};
\node[below] at (4.25,-0.75) {$s'$};
\end{tikzpicture}
\end{equation}
The unit square ${\tilde\epsilon}_V({u})$ of the vertical structure at the horizontal arrow $ u$ is the square
\begin{equation} \label{dg4a}
\begin{tikzpicture}[baseline=(current  bounding  box.center)]
\foreach \x in {0}
{  \draw[thick,stealth'-] (1.5+\x,0) -- (0+\x,0);
  \draw[thick,stealth'-] (1.5+\x,1.5) -- (0+\x,1.5);
  \draw[thick,-stealth'] (1.5+\x,0) -- (1.5+\x,1.5);
  \draw[thick,-stealth'] (0+\x,0) -- (0+\x,1.5);}
\node[left] at(-0.,0.75) {$ {{\epsilon}_V({ \alpha}_H({ u}))}$};
\node[above] at(0.75,1.5) {$ u$};
\node[below] at(0.75,-0.1) {$ u$};
\node[right] at (1.5,0.75) {${ {\epsilon}_V({ \beta}_H({ u}))}$};
\end{tikzpicture}
\end{equation}

The definition of a double groupoid is completed by the verification of a compatibility condition between the two structures, namely,
\begin{equation} \label{dg5}
({\mathbb A}\odot_H{\mathbb B})\odot_V({\mathbb C}\odot_H{\mathbb D})=({\mathbb A}\odot_V{\mathbb C})\odot_H({\mathbb B}\odot_V{\mathbb D}),
\end{equation}
whenever the operations are possible. Graphically, this condition means that, in a square made of four squares whose edges in contact are matched, the result is the same whether one first composes horizontally and then vertically, or vice versa. In other words, the large square shown below makes sense.
\begin{equation} \label{dg6}
\begin{tikzpicture}[baseline=(current  bounding  box.center)]
\foreach \x in {0,1.5}
{\foreach \y in {0,1.5}
{  \draw[thick,stealth'-] (1.5+\x,0+\y) -- (0+\x,0+\y);
  \draw[thick,stealth'-] (1.5+\x,1.5+\y) -- (0+\x,1.5+\y);
  \draw[thick,-stealth'] (1.5+\x,0+\y) -- (1.5+\x,1.5+\y);
  \draw[thick,-stealth'] (0+\x,0+\y) -- (0+\x,1.5+\y);}}
\node at (0.75,0.75) {$\mathbb D$};
\node at (0.75,2.25) {$\mathbb B$};
\node at (2.25,0.75) {$\mathbb C$};
\node at (2.25,2.25) {$\mathbb A$};
\end{tikzpicture}
\end{equation}

\begin{rem} {\rm The {\it double source map} assigning to each square $\mathbb A$ its sources in the horizontal and vertical structures, ${\tilde \alpha}_H({\mathbb A}), {\tilde \alpha}_V({\mathbb A})$, is often assumed to be surjective. The graphical representation of this assumption is that a right angle made of a horizontal and a vertical arrow issuing from a common corner in $\mathcal B$ can always be completed (not necessarily uniquely) to a square in $\mathcal Q$. This is known as the {\it filling condition}.
}
\end{rem}

\begin{example} {\rm A useful example of a double groupoid is the {\it coarse double groupoid} $\square({\mathcal H},{\mathcal V})$, generated by two groupoids, $\mathcal H$ and $\mathcal V$, with a common base set. It consists of all the possible squares that can be formed using the two groupoids as sides. Any double groupoid can be naturally mapped as an inclusion into the coarse double groupoid generated by its side groupoids.
}
\end{example}

\begin{definition}\label{def:core} {\rm 
The {\it core} ${\mathcal C}_{\mathcal Q}$ of a double groupoid $\mathcal Q$ is an ordinary groupoid consisting of the collection of all squares $\mathbb A$ for which the two source arrows $s$ and $\hat s$ are, respectively, the units of $\mathcal H$ and $\mathcal V$ at the point $X=\alpha_H\left(\hat{\alpha}_V({\mathbb A})\right)=\alpha_V\left(\hat{\alpha}_H({\mathbb A})\right)$.
}
\end{definition}

A square in the core groupoid, therefore, will have 3 coincident vertices, $W_1=W_2=W_3$, as represented in Figure \ref{fig:coredouble}. For this reason, the core can be given a canonical groupoid structure by regarding its arrows as pairs of arrows $(t,{\hat t})$ with common source $X=W_0$ and target point $Z=\beta_H\left(\hat{\beta}_V({\mathbb A})\right)=\beta_V\left(\hat{\beta}_H({\mathbb A})\right)= W_3$. The composition  law of two squares in the core groupoid (with non-trivial sides $t, {\hat t}$ and $t', {\hat t}'$) is the square with non-trivial sides $tt'$ and ${\hat t} {\hat t}'$. The identity of the core groupoid at a point $b \in {\mathcal B}$ is the square in ${\mathcal Q}$ all of whose sides are identities (two identities of $\mathcal H$ and two identities of $\mathcal V$ at $b$). In many respects the core groupoid characterizes the double groupoid \cite{brownmac}, although the physical meaning of the original side groupoids is not preserved when passing to the core.
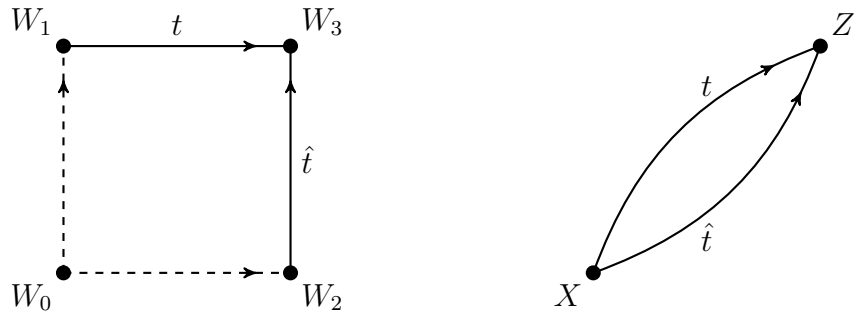
\begin{figure}[H]
\begin{center}
\begin{tikzpicture} [scale=1.0]
\tikzset{->/.style={decoration={
  markings,
  mark=at position .85 with {\arrow{stealth'}}},postaction={decorate}}}
\draw[dashed, thick, ->] (0,0)--(3,0);
\draw[dashed, thick,->] (0,0)--(0,3);
\draw[thick,->] (3,0)--(3,3);
\draw[thick,->] (0,3)--(3,3);
\node at (0,0)[circle,fill,inner sep=2pt]{};
\node at (0,3)[circle,fill,inner sep=2pt]{};
\node at (3,0)[circle,fill,inner sep=2pt]{};
\node at (3,3)[circle,fill,inner sep=2pt]{};
\node[right] at (3,1.5) {$\hat t$};
\node[above] at (1.5,3) {$t$};

\node [below left] at (0,0) {$W_0$};
\node [above left] at (0,3) {$W_1$};
\node [below right] at (3,0) {$W_2$};
\node [above right] at (3,3) {$W_3$};

\begin{scope} [xshift= 7cm]
\node at (0,0)[circle,fill,inner sep=2pt]{};
\node at (3,3)[circle,fill,inner sep=2pt]{};
\draw [thick,->] (0,0) to [bend left=25] (3,3);
\node[above] at (1.5,2.2) {$t$};
\draw [thick,->] (0,0) to [bend left=-25] (3,3);
\node[below] at (1.5,0.8) {$\hat t$};
\node[below left] at (0,0) {$X$};
\node[above right] at (3,3) {$Z$};
\end{scope}

\end{tikzpicture}
\end{center}
\caption{An element in the core, where identity arrows are dashed.}
\label{fig:coredouble}
\end{figure}

\subsection{Double groupoids of composites}

Working out the case of binary weighted composites, we start from the horizontal and vertical material groupoids, ${\mathcal H} \rightrightarrows {\mathcal B}$ and ${\mathcal V} \rightrightarrows{\mathcal B}$, associated, respectively, with the material response functions ${\bf f}_1$ and ${\bf f}_2$ of the two constituents of the composite $\mathcal B$, as described in Section \ref{sec:weighted}. We note that each of these groupoids draws its arrows from the general linear group $GL(3, {\mathbb R})$.\footnote{More technically, these are subgroupoids of the 1-jet groupoid over $\mathcal B$.} We can, therefore, define the total set $\mathcal Q$ of the double groupoid of the composite as the subset of squares in the coarse double groupoid $\square({\mathcal H}, {\mathcal V})$ that satisfy the commutation condition.
\begin{equation}
{\hat t}\, s = t\,{\hat s}.
\end{equation}

Consider now the core ${\mathcal C}_{\mathcal Q}$ of the material double groupoid ${\mathcal Q}\rightrightarrows {\mathcal B}$ of a weighted binary composite $\mathcal B$. 

\begin{prop} {\rm A binary composite satisfies the condition of weighted uniformity if, and only if, the core of its material double groupoid is transitive.
}
\end{prop}
\begin{proof} Since every square of $\mathcal Q$ is commutative, it follows from Definition \ref{def:core} that every pair $(t, {\hat t}) \in {\mathcal C}_{\mathcal Q}$ satisfies the condition $t={\hat t}$. The proposition is an immediate consequence of Definition \ref{def:weightedunif}
\end{proof}

\section{Including the volume fractions}
\label{sec:volumefractions}

\subsection{The bundle of volume fractions}

In the case of a binary composite, the material composition at a point $X \in {\mathcal B}$ can be completely characterized by the ratio
\begin{equation}
r(X)=\frac{k_1(X)}{1-k_1(X)},
\end{equation}
a ratio which falls always in the open interval $(0, \infty)$. For symmetry, it is sometimes convenient to use the logarithmic measure $s(X)=\log{r(X)}$, which covers the whole real line $\mathbb R$. 

The material double groupoid $\mathcal Q$ introduced in Section \ref{sec:materialdouble} makes no reference to this structural ratio. A weighted binary composite with given ratio $s(X)$ can be considered as a section of the trivial fibre bundle ${\mathcal C}={\mathcal B} \times {\mathbb R}$, which we will call the bundle of volume fractions. We remark that in the presence of such phenomena as allotropic transformations, where two different crystal structures of the same material coexist, we can have a time evolution of the filed of structural ratios, which will be reflected in a time variation of the cross section of $\mathcal C$. In the case of $n$-ary composites with $n>2$, true chemical reactions can occur that will lead to similar time evolutions of an $n$-fold groupoid.

When looking for an entity that may serve as a mathematical support for such phenomena, it is clear that the material double groupoid $\mathcal Q$, in an of itself, will not provide this support. Indeed, to reflect the evolving structural ratios, the squares of $\mathcal Q$ will have to be endowed with additional parameters (most obviously, the structural ratios of the four corners, although other possibilities might be considered).

\subsection{Weighted groupoids}

Let $\mathcal B$ be a manifold and let $\pi: {\mathcal C} \to {\mathcal B}$ be a fibre bundle over $\mathcal B$ with projection map $\pi$. The typical fibre $\bf C$ need not be a group, but just represent some extra degrees of freedom whose meaning depends on the context. If ${\mathcal G} \rightrightarrows {\mathcal B}$ is a groupoid on $\mathcal B$ with source and target projections $\alpha,\beta:{\mathcal G} \to {\mathcal B}$, we define the extended groupoid ${\mathcal G}_C \rightrightarrows {\mathcal C}$ by declaring every arrow from $c_1 \in {\mathcal C}$ to $c_2 \in {\mathcal C}$ as a triple $(u, c_1,c_2)$, where $u\in {\mathcal G}$ with $\alpha(u)=\pi(c_1)$ and $\beta(u)=\pi(c_2)$. The source and target maps in ${\mathcal G}_C$ are simply defined as $\alpha_C((u,c_1,,c_2))=c_1$ and $\beta_C((u,c_1,,c_2))=c_2$. In a given physical context described by the goupod ${\mathcal G} \rightrightarrows {\mathcal B}$, the fibres may represent an additional internal variable not contemplated in the original groupoid.

Thus, the arrows of ${\mathcal G}_C$ can be regarded as precise lifts of the arrows of $\mathcal G$. The unit at $c \in {\mathcal C}$ is the lift of the unit at $\pi(c) \in {\mathcal B}$. Composition of two arrows in ${\mathcal G}_C$ is possible only if the tip-to-tail condition is satisfied for the three entries of the corresponding triples.. In view of our intended application we will refer to  ${\mathcal G}_C$ as a weighted groupoid. It should be noted that even if there is just one arrow in $u \in \mathcal G$ with source $b_1$ and target $b_2$, the collection of its lifted versions in ${\mathcal G}_C$ will be in a one-to-one correspondence with ${\bf C}^2$. Some of the features of the extended groupoid ${\mathcal G}_C$ are represented in Figure \ref{fig4} .

\begin{figure}[H]
\begin{center}
\begin{tikzpicture} [scale=1.]
\tikzset{->/.style={decoration={
  markings,
  mark=at position .95 with {\arrow{stealth'}}},postaction={decorate}}}
\draw[thick] (0,0) to [bend left=15] (10,0);
    \node[below right] at (10,0) {$\mathcal B$};
    \foreach \angle in {108, 102, ..., 72} {
        \draw  ([shift=(\angle:16.25) ](5,-15)--([shift=(\angle:20)](5,-15);}
 \draw[->] (0.15,0.06) to [bend left] (1.70,0.43);
  \draw[->] (1.70,0.43) to [bend left] (3.40,0.70);
  \node[below] at (0.15,0.06) {$b_1$};
   \node[below] at (1.70,0.43) {$b_2$};
      \node[below] at (3.40,0.70) {$b_3$};
   \draw[thick,->] (-0.5,2) to [bend left=45] (1.15,3.1);
   \draw[thick,->] (1.15,3.1) to [bend left=30] (3.,3.5);
   \draw[dashed, ->] (-0.5,2) to [bend left] (0.15,0.06);
   \node [right] at (0.1,1.2) {$\pi$};
   \node [left] at (-0.5,2) {$c_1$};
   \node [above] at (0.83,0.40) {$u$};
   \node[above, rotate=30] at (0.3,2.9) {$_{(u,c_1,c_2)}$};
   \node[below right] at (1.15,3.25) {$c_2$};
   \node[right] at (3.0,3.5) {$c_3$};
   \node[below] at (6.60,0.70) {$b$};
\draw [->] (6.60,0.70) arc (-90:250:0.2); 
\node[above right] at (6.73,  0.70) {$\epsilon(b)$};
\draw [thick, ->] (6.85,3) arc (0:290:0.3);
\draw[thick] (6.55,2.7) arc (270:359:0.3);
\node [right] at (6.85,2.95) {$c$};
\draw[thick, ->] (6.85,3) to [bend right] (7.04,4.5);
\node [right] at  (7.04,4.5) {$c'$};
\node[right] at (7.1,3.75) {$_{(\epsilon(b), c, c')}$};
\node [above] at (6.2, 3.2) {$_{\epsilon_C(c)}$};
\node [above] at (7.1,5.) {${\mathcal C}_b$};
\end{tikzpicture}
\end{center}
\captionsetup{justification=centering}
\caption{Groupoid on $\mathcal B$ extended to a bundle ${\mathcal C}\to {\mathcal B}$. \\ The unit at $c$ is the triple $\epsilon_C(c)=(\epsilon(b),c,c)$.}
\label{fig4}
\end{figure}
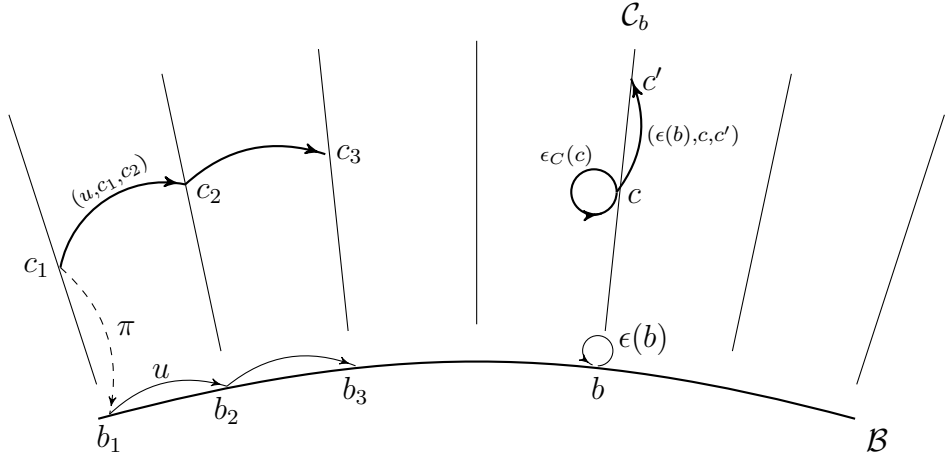

\begin{rem} {\rm This construction is most probably equivalent to the Cartesian product of the groiupoid ${\mathcal G}\rightrightarrows{\mathcal B}$ times the pair groupoid ${\bf C}\times {\bf C} \rightrightarrows {\bf C}$.
}
\end{rem}

\subsection{The weighted material double groupoid of a composite}

Given the underlying material groupoids ${\mathcal H}\rightrightarrows {\mathcal B}$ and ${\mathcal V} \rightrightarrows{\mathcal B}$ of the two constituents of a composite, we construct their extended groupoids ${\mathcal H}_R\rightrightarrows{\mathcal B}\times{\mathbb R}$ and ${\mathcal V}_R\rightrightarrows {\mathcal B} \times {\mathbb R}$ as described above. Using these as side groupoids, we obtain the weighted material double groupoid ${\mathcal Q}_R$ of the weighted binary composite. A schematic description can be gathered from Figure \ref{fig5}. By analogy with the terminology of \cite{natale}, we designate the original, non-weighted, double groupoid $\mathcal Q$ as the {\it frame} of ${\mathcal Q}_R$. The collection of squares in ${\mathcal Q}_R$ hovering over a single square in the frame is in a one-to one correspondence with ${\mathbb R}^4$. This feature is reminiscent of, but not equivalent to, the concept of non-slim double groupoid, which is a much more sophisticated construct. Indeed, as defined, the new groupoid ${\mathcal G}_C$ is technically still a slim groupoid, but from the point of view of an ``observer'' unaware of the existence of the fibres, the squares of the new groupoid are not uniquely determined by their (projected) sides. 

\begin{figure}[H]
\begin{center}
\begin{tikzpicture} [scale=1.0]
\tikzset{->/.style={decoration={
  markings,
  mark=at position .95 with {\arrow{stealth'}}},postaction={decorate}}}
  \draw[fill=gray!20] (0,0.2) -- (10,0.2) --(12.4,3.2)--(2.4,3.2)--cycle;
  \foreach \x in {0,2,...,8}
  \foreach \y in {0,...,2}
  \draw (\x-0.5+2+0.8*\y,0.9+\y)--(\x-0.5+2+0.8*\y,5+\y);
  \foreach \x in{0,2}
  \foreach \y in {0,1}
  \draw[->] (3.5+\x+0.8*\y,0.5+\y) to [bend left] (5.5+\x+0.8*\y,0.5+\y);
 \foreach \x in{0,2,4}
    \draw[->] (3.5+\x,0.5) to [bend left=25] (4.3+\x,0.5+1);
    \draw[->, thick]  (3.5,3.5) to  [bend left] (5.5,4);
    \draw[->, thick] (5.5,4) to[bend left] (6.3,5);
    \draw[->, thick]  (3.5,3.5) to  [bend left=15] (4.3,5.5);
    \draw[->, thick]  (4.3,5.5) to  [bend left=15] (6.3,5);
    \draw[->, thick]  (6.3,5) to  [bend left] (8.3,5.75);
    \draw[->, thick]  (5.5,4) to  [bend left] (7.5,3.5);
    \draw[->, thick]  (7.5,3.5) to  [bend left=15] (8.3,5.75);
    \node at ((9.8,0.5) {$\mathcal B$};
\end{tikzpicture}
\end{center}
\caption{Squares in a weighted double groupoid and in its frame}
\label{fig5}
\end{figure}
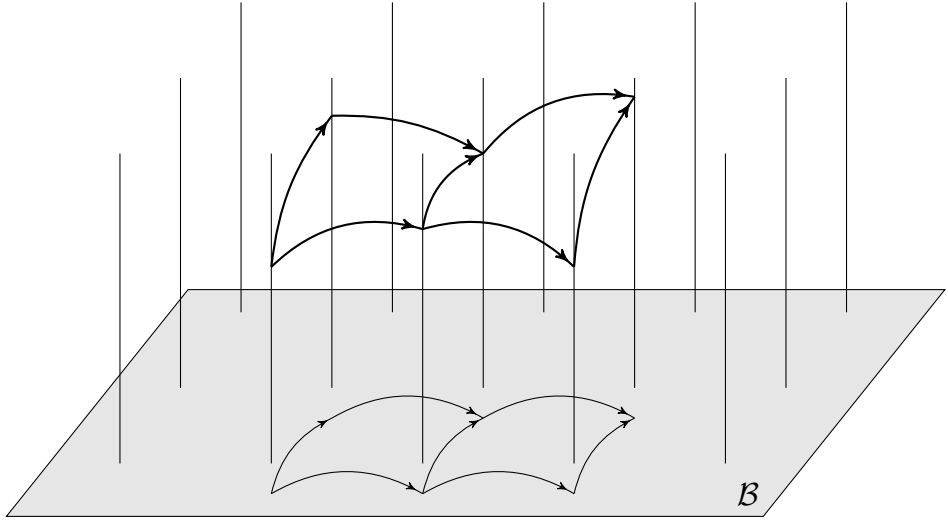

\subsection{Geometric picture of material evolution}
\label{sec:evolution}

If it is assumed that a material body preserves its point-wise material properties with the passage of time, what is the most general possible kind of material evolution that such a body can undergo? Moreover, what are real-life phenomena that exemplify this kind of evolution? The most common, but not the only, examples of this sort arise in the field of biological tissues as they undergo processes of growth and remodelling. In the case of remodelling, no new material is generated, but the existing material re-accommodates itself under the influence of internal or external demands. In the familiar orthodontic procedure, for instance, the trabeculae of the spongy (cancellous) bone alter their orientation to more efficiently carry the applied forces, eventually producing the desired effect.

If the body is materially uniform and made of a single component, remodelling changes manifest themselves in alterations of the material groupoid. More graphically, the arrows of the material groupoid are functions of time, always remaining within the general linear group. These changes stem from the right action of this group on the constitutive function. In a composite, on the other hand, the arrows of each side groupoid will evolve in general along different paths, that is, by the right action of different time-dependent elements. For example, one of the constituents may not evolve at all while another may undergo a non-trivial evolution. This non-harmonized evolution may be the source of increasing or decreasing misalignment between the constituents. A cartoon of such process would be a collection of pairs of earthworms wiggling as they try to free themselves from restraining pins placed at both their ends.

In the case of a weighted composite a new degree of freedom is available to the evolution process, namely, the volume fractions. How is this phenomenon reflected in our new geometric entity, that is, the extended groupoid ${\mathcal G}_C\rightrightarrows{\mathcal B}$? While the arrows are busy wiggling (horizontally, as it were) within the frame ${\mathcal G}$, the (smooth) section of the fibre bundle moves (vertically, as it were) according to laws of evolution that may, or may not, be coupled with the evolution of the arrows below.

\section {Non-slim double groupoids}
\label{sec:nonslim}

Although not explicitly stated, the concept of double groupoid presented in Section \ref{sec:double} presupposes that every square in the double groupoid is completely defined by its sides. Indeed, the composition laws represented in Equations (\ref{dg3}) and (\ref{dg4}) stipulate outcomes defined exclusively in terms of their side arrows. In the terminology of \cite{natale} a double groupoid is said to be {\it slim} if its squares are completely determined by their sides.

The structure of non-slim double groupoids, namely double groupoids in which more squares that one have the same sides, has been completely determined (at least for the discrete case) in the original study of \cite{natale}. One of the main results of that study is that the extra degree of freedom available in a non-slim double groupoid is controlled by an Abelian group bundle $\mathcal K$ over the base set $\mathcal B$, with typical fibre $\bf K$.

Given any non-slim double groupoid $\mathcal Q$ the collection $\mathcal F$ of all the squares in the coarse double groupoid $\square({\mathcal H},{\mathcal V})$ determined by the sides of the squares in $\mathcal Q$ is a (slim) double groupoid called the {\it frame} of $\mathcal Q$. The map $\Pi:{\mathcal Q} \to {\mathcal F}$ is non-injective, since $\mathcal Q$ is not slim. A cross section $\mu: {\mathcal F} \to {\mathcal Q}$ assigns to each square $\mathbb A$ in the frame a square ${\mathbb A}'=\mu({\mathbb A})$ in $\mathcal Q$ such that  $\Pi({\mathbb A}')={\mathbb A}$.

The extensions of slim double groupoids into non-slim ones rigorously introduced in \cite{natale} contemplate in general that the composition operations between squares in the non-slim groupoid may involve an entanglement, as it were, of the composition operations in both the Abelian group bundle $\mathcal K$ and the frame $\mathcal F$. The simplest, always available, option is to consider $\mathcal K$ as the product bundle ${\mathcal B}\times {\bf K}$ and the non-slim groupoid as the Cartesian product ${\mathcal Q}={\mathcal F} \times {\bf K}$. Simply put, the squares in $\mathcal Q$ are ordered pairs of a square in the frame and an element in $\bf K$. Moreover, the composition of two squares of the non-slim groupoid is the pair consisting of the composition of the corresponding squares in the frame, and the sum of the corresponding values of the factors.

 If we assume that in our application the typical fibre of $\bf K$ is identified with the group ${\mathbb R}^4$, it is clear that the addition of structural ratios at the corners conveys no physical meaning. It is for this reason that the idea of using non-slim extensions of the material double bundle of a composite will be abandoned for now in favour of a less restrictive kind of generalization that still keeps some of the flavour of non-slim groupoids, most significantly the property that the squares of the generalized double groupoid $\mathcal Z$ are not uniquely determined by the sides of the corresponding squares in $\mathcal Q$. Nevertheless, in Section \ref{sec:fees} the non-slim paradigm of \cite{natale} will be shown to be useful for a different application.

\subsection{A model of transplant costs}
\label{sec:fees}

In this section, we offer a potential application of the extension of a slim double groupoid introduced by \cite{natale}. For the sake of clarity, let us start from a single (uniform) constituent (a 1-ary composite, if you will). Let us now suppose that each arrow $u$ in the material groupoid ${\mathcal Z} \rightrightarrows {\mathcal B}$, that is, each material isomorphism, is an actual operation of transplant whereby a small material element is cut around $s=\alpha(u)$ and then implanted at $t=\beta(u)$. This operation carries two energetic costs (or benefits), a cost of harvesting at $s$ and a cost of implanting at $t$. In this way, each point $b$ of $\mathcal B$ is assigned two real parameters, $\omega(b)$ and $\chi(b)$, respectively, to be associated to each departing or arriving arrow at $b$. We assume the range of these parameters to be the real line.

Let us refer to the original groupoid ${\mathcal Z} \rightrightarrows {\mathcal B}$ as the {\it material frame} groupoid. We intend to build an extension ${\mathcal Q} \rightrightarrows{\mathcal B}$ of this groupoid that takes into consideration the transplant costs just described. The arrows in $\mathcal Q$ will be denoted with capital letters, and their composition with the operator $\circ$. Each arrow $u \in {\mathcal Z}$, when assigned some corner weights $\omega(\alpha(u))$ and $\chi(\beta(u))$, gives rise to an arrow $U \in {\mathcal Q}$. For notational economy let us denote $\omega(U)=\omega(\alpha(u))$ and $\chi(U)=\chi(\beta(u))$. Consider two composable arrows $u, v \in{\mathcal Z}$. Their composition $vu$ results in an arrow of $\mathcal Z$  such that $\alpha(vu)=\alpha(u)$ and $\beta(vu)=\beta(v)$. The costs of this arrow are $\omega(\alpha(u))$ and $\chi(\beta(v))$. On the other hand, if we take into consideration the additive nature of costs of harvesting and implanting, we have the pairs 
\begin{equation}
U=[u, (\omega(U), \chi(U))],
\end{equation}
and
\begin{equation}
V=[v, (\omega(V), \chi(V))],
\end{equation}
whose composition in $\mathcal Q$ is
defined as
\begin{equation}
V \circ U=[vu, (\omega(U)+\omega(V), \chi(U)+\chi(V))].
\end{equation}

Since there is an infinity of ways to arrive from a source to a target of an implant (by breaking it into different partial processes) we conclude that to one and the same arrow in the frame $\mathcal Z$ we can associate an infinity of vectors (with the same source and target) in the extended groupoid $\mathcal Q$. This feature is illustrated schematically in Figure \ref{fig6}. Notice that, contrary to what is depicted in Figure \ref{fig4}, the tip-to-tail condition appears not to be satisfied in $\mathcal Q$. Nevertheless, the base of the extended groupoid remains always the base $\mathcal B$, and the source and target maps are the same for both groupoids $\mathcal Z$ and $\mathcal Q$. 

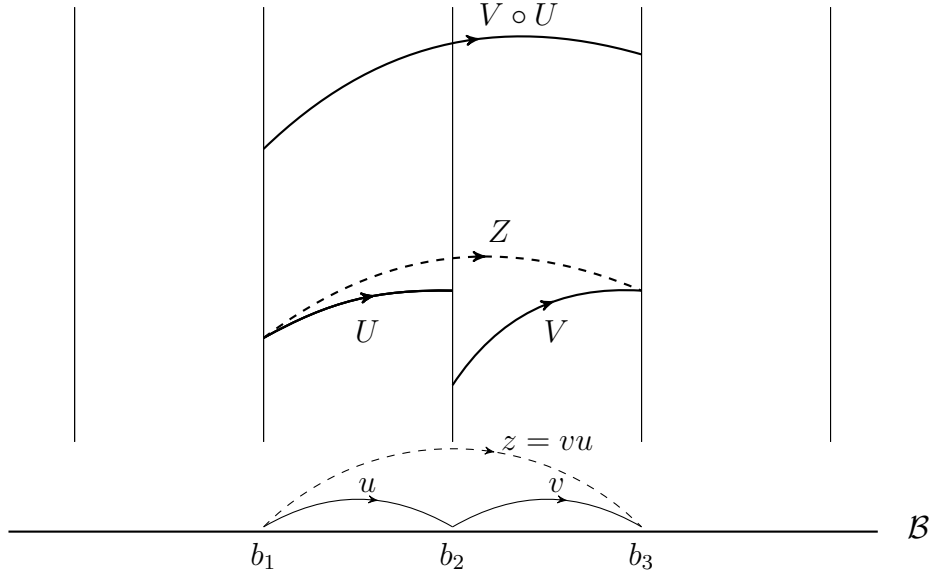
\begin{figure}[H]
\begin{center}
\begin{tikzpicture} [scale=1.25]
\tikzset{->/.style={decoration={
  markings,
  mark=at position 0.6 with {\arrow{stealth'}}},postaction={decorate}}}
\draw[thick] (0.8,0.45)-- (10,0.45);
  \foreach \x in {0,2,...,8}
  \foreach \y in {0,0}
  \draw (\x-0.5+2+0.8*\y,1.4+\y)--(\x-0.5+2+0.8*\y,6+\y);
  \foreach \x in{0,2}
   \draw[->] (3.5+\x,0.5) to [bend left] (5.5+\x,0.5);
   \draw[dashed,->] (3.5,0.5) to [bend left=45] (7.5,0.5);
 \foreach \x in{0,2,4}
    \draw[->, thick]  (3.5,2.5) to  [bend left=15] (5.5,3);
    \draw[->, thick]  (5.5,2) to  [bend left] (7.5,3.);
    \draw[thick, ->] (3.5,2.5+2) to [bend left]  (7.5, 2.5+3.);
    \draw[dashed, thick, ->] (3.5,2.5) to  [bend left] (7.5,3);
    \node[right] at ((10.2,0.5) {$\mathcal B$};
    \node[above] at (4.6,0.75) {$u$};
        \node[above] at (6.6,0.75) {$v$};
  \node[above right] at (5.9,1.2) {$z=vu$};
  \node[below] at (3.5,0.45) {$b_1$};
    \node[below] at (5.5,0.45) {$b_2$};
      \node[below] at (7.5,0.45) {$b_3$};
      \node[below] at (4.6,2.8) {$U$};
            \node[below] at (6.6,2.8) {$V$};
                  \node[above] at (6.,3.4) {$Z$};
                  \node[above] at (6.2,5.7) {$V \circ U$};
             
\end{tikzpicture}
\end{center}
\caption{Extension of a groupoid on $\mathcal B$}
\label{fig6}
\end{figure}

This model can be extended also to a double groupoid. Assuming a binary composite with spatially constant volume fractions, the transplant cost (assuming it has a concrete meaning) should be based on the consideration of both components, particularly if the composite is not uniform. Given a square $\mathbb A$ in the frame $\mathcal Q$, and considering that each node can be assigned two weights, the set of squares sharing the sides of $\mathbb A$ is in a one-to-one correspondence with ${\mathbb R}^8$. It can be verified that the resulting structure is a non-slim double groupoid in the sense of \cite{natale}. A schematic representation is given in Figure \ref{fig7}.

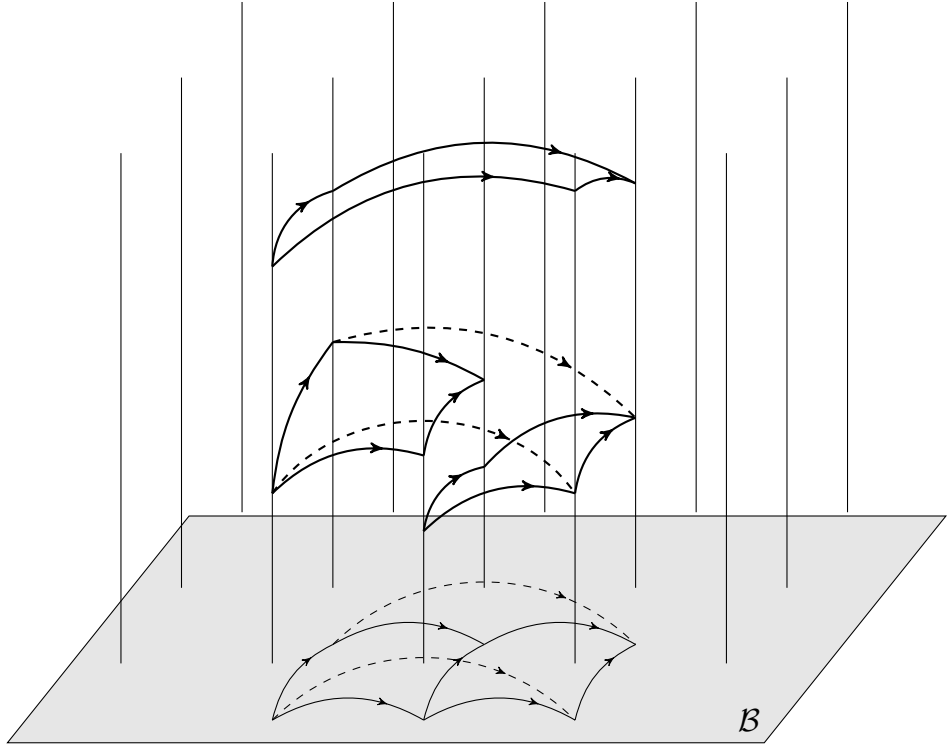
\begin{figure}[H]
\begin{center}
\begin{tikzpicture} [scale=1.0]
\tikzset{->/.style={decoration={
  markings,
  mark=at position .75 with {\arrow{stealth'}}},postaction={decorate}}}
  \draw[fill=gray!20] (0,0.2) -- (10,0.2) --(12.4,3.2)--(2.4,3.2)--cycle;
  \foreach \x in {0,2,...,8}
  \foreach \y in {0,...,2}
  \draw (\x-0.5+2+0.8*\y,1.25+\y)--(\x-0.5+2+0.8*\y,8+\y);
  \foreach \x in{0,2}
  \foreach \y in {0,1}
  \draw[->] (3.5+\x+0.8*\y,0.5+\y) to [bend left] (5.5+\x+0.8*\y,0.5+\y);
 \foreach \x in{0,2,4}
    \draw[->] (3.5+\x,0.5) to [bend left=25] (4.3+\x,0.5+1);
    \draw[->, thick]  (3.5,3.5) to  [bend left] (5.5,4);
    \draw[->, thick] (5.5,4) to[bend left] (6.3,5);
    \draw[->, thick]  (3.5,3.5) to  [bend left=15] (4.3,5.5);
    \draw[->, thick]  (4.3,5.5) to  [bend left=15] (6.3,5);
    \draw[->, thick]  (6.3,3.85) to  [bend left] (8.3,4.5);
    \draw[->, thick]  (5.5,3) to  [bend left] (7.5,3.5);
   \draw[->, thick]  (5.5,3) to  [bend left=35] (6.3,3.85);
    \draw[thick,->] (7.5,3.5) to  [bend left] (8.3,4.5);
    \draw[thick,->] (3.5, 6.5) to [bend left] (7.5,7.5);
    \draw[thick,->] (3.5,6.5) to  [bend left=35] (4.3, 7.5);
        \draw[thick,->] (4.3, 7.5) to [bend left] (8.3,7.6);
            \draw[thick,->] (7.5, 7.5) to [bend left] (8.3,7.6);
    \node at ((9.8,0.5) {$\mathcal B$};
    \draw[dashed, ->] (3.5,0.5) to [bend left=45] (7.5,0.5);
        \draw[dashed, ->] (4.3,1.5) to [bend left=45] (8.3,1.5);
        \draw[thick,dashed,->] (3.5,3.5) to [bend left=55] (7.5,3.5);
                \draw[thick,dashed,->] (4.3,5.5) to [bend left] (8.3,4.5);
\end{tikzpicture}
\end{center}
\caption{Squares in a non-slim double groupoid and in its frame}
\label{fig7}
\end{figure}

\section{Conclusions}

A mathematical framework has been proposed to provide a unified treatment of the mechanics of a large class of composite materials. Although the treatment has been confined to the case of binary composites, many of the features can be extended to the $n$-ary case for $n>2$. The main mathematical structure invoked is the theory of double groupoids which was previously shown to fit the description of binary composites and to characterize and quantify the presence of defects of misalignment between the constituents. This theoretical scaffold has been shown to be extendable to composites of spatially variable volume fractions, thus opening the way to the eventual inclusion of chemical reactions between the constituents.


\newpage
\begin{thebibliography}{99}

\bibitem{natale} Andruskiewitsch N and Natale S (2009), The structure of double groupoids, {\it Journal of Pure and Applied Algebra} {\bf 213}, 1031--1045.

\bibitem{bowen1} Bowen R M (1967), Toward a Thermodynamics and Mechanics of Mixtures, {\it Archive for Rational Mechanics and Analysis} {\bf 24}, 370--403.

\bibitem{bowen2} Bowen R M (1976), Theory of Mixtures, {\it in Continuum Physics, III, Edited by A.C. Eringen}, 1--127.

\bibitem{brandt} Brandt H (1927), \H{U}ber eine Verallgemeinerung des Gruppenbegrieffes, {\it Mathematische Annalen} {\bf 96}, 360--366.

\bibitem{brownmac} Brown R and Mackenzie K C H (1992), Determination of a double Lie groupoid by its core diagram, {\it Journal of Pure and Applied Algebra} {\bf 80}, 237--272.

\bibitem{symmetry} Epstein M (2021), Material geometry of binary composites, {\it Symmetry} {\bf 13}/5, 892; https://doi.org/10.3390/sym13050892.

\bibitem{jmps} Epstein M (2022), The theory of continuous distributions of composite defects. {\it Journal of the Mechanics and Physics of Solids} {\bf 167}, 104983. 

\bibitem{book} Epstein M (2026) {\it An Elementary Introduction to n-fold Groupoids}, Springer.

\bibitem{eandmdel} Epstein M and de Le\'on M (2000), Homogeneity without uniformity: toward a mathematical theory of functionally graded materials, {\it International Journal of Solids and Structures} {\bf 37}, 7577--7591.

\bibitem{eande} Epstein M and El\.zanowski M Z (2025), Material evolution of defective composites
{\it Mathematics and Mechanics of Solids} {\bf 30}/2, 202--210.

\bibitem{victor} Jim\'enez VM, de Le\'on M and Epstein M (2025), Double groupoids of composites: applications to uniformity, {\it Geometric Mechanics}, 10.1142/S297245892650005X, 1--29.

\bibitem{noll} Noll W (1967), Materially Uniform Bodies with Inhomogeneities, {\it Archive for Rational Mechanics and Analysis} {\bf 27}, 1-32.

\bibitem{weinstein} Weinstein A (2000), Unifying internal and external symmetry, {\it Notices of the American Mathematical Society} {\bf 47}/7, 744--752.

\end{thebibliography}
\end{document}